\documentclass[american,aps,prl,reprint,superscriptaddress,showpacs,footinbib]{revtex4-1}
\usepackage[T1]{fontenc}
\usepackage{color}
\usepackage{babel}
\usepackage{palatino}
\usepackage{mathpazo}

\usepackage{amsmath,amssymb,amsthm,verbatim}
\usepackage{tikz}
\usetikzlibrary{arrows}
\usepackage{palatino}
\usepackage{mathpazo}
\usepackage{stmaryrd}
\usepackage{mathtools}

\usepackage{enumitem}
\setlist{nosep}

\newtheorem{theorem}{Theorem}
\newtheorem{definition}[theorem]{Definition}
\newtheorem{lemma}[theorem]{Lemma}
\newtheorem{proposition}[theorem]{Proposition}

\newtheorem{corollary}[theorem]{Corollary}

\newcommand{\mbb}{\mathbb}
\newcommand{\mc}{\mathcal}
\newcommand{\msf}{\mathsf}
\newcommand{\mrm}{\mathrm}

\newcommand{\norm}[1]{\left|\left|#1\right|\right|_1}

\newcommand{\bra}[1]{\langle #1 |}
\newcommand{\ket}[1]{| #1 \rangle}
\newcommand{\op}[2]{|#1\rangle \langle #2|}
\newcommand{\ip}[2]{\langle #1|#2\rangle}

\newcommand{\tr}{\mathrm{Tr}}
\newcommand{\pr}{\mathrm{Pr}}

\newcommand{\apxep}{\overset{\epsilon}{\approx}}

\newcommand{\wt}[1]{\widetilde{#1}}

\begin{document}

\title{Dephasing-Covariant Operations Enable Asymptotic Reversibility of Quantum Resources}

\author{Eric Chitambar}

\email{echitamb@siu.edu}

\affiliation{Department of Physics and Astronomy, Southern Illinois University,
Carbondale, Illinois 62901, USA}

\begin{abstract}
We study the power of dephasing-covariant operations in the resource theories of coherence and entanglement.  These are quantum operations whose actions commute with a projective measurement.  In the resource theory of coherence, we find that any two states are asymptotically interconvertible under dephasing-covariant operations.  This provides a rare example of a resource theory in which asymptotic reversibility can be attained without needing the maximal set of resource non-generating operations.  When extended to the resource theory of entanglement, the resultant operations share similarities with LOCC, such as prohibiting the increase of all R\'{e}nyi $\alpha$-entropies of entanglement under pure state transformations.  However, we show these operations are still strong enough to enable asymptotic reversibility between any two maximally correlated mixed states, even in the multipartite setting.

\end{abstract}

\maketitle

A quantum resource theory (QRT) studies the dynamics of physical systems under a restricted class of operations.  The defining components in any QRT are a set of allowable or ``free'' operations $\mc{O}$ and a collection of ``free'' states $\mc{F}$.  To capture the notion of resource, the free operations are required to act invariantly on the set of free states, and any state not belonging to $\mc{F}$ is said to possess resource and called a resource state.  For example, in entanglement theory the set $\mc{O}$ typically consists of local quantum operations and classical communication (LOCC), and the free states $\mc{F}$ are non-entangled or separable states \cite{Plenio-2007a}.  In general, the set $\mc{F}$ does not uniquely determine the free operations, and different QRTs can have the same set of free states.  For instance, in the recently proposed resource theories of quantum coherence, the free states are density matrices that are diagonal in some specified basis, but a number of distinct operational classes have been proposed that all share in the property of preserving the diagonal states \cite{Aberg-2006a, Baumgratz-2014a, Winter-2015a, Chitambar-2016c, Yadin-2016a, deVicente-2017a}.

A fundamental question considered in any QRT is whether one state can be converted to another using the free operations.  If it is possible to transform $\rho\to \sigma$ with the allowed operations of the QRT, then one is well-justified in concluding that $\rho$ contains no smaller amount of resource than $\sigma$.  Unfortunately, most pairs of states will lack such a convertibility relationship within a given QRT, including entanglement and coherence theories.  One alternative is to consider \textit{asymptotic transformations} in which one allows for multiple copies of the source/target states and relaxes the condition of perfect transformation.  

For a QRT $(\mc{F},\mc{O})$ with free states and free operations defined on any finite-dimensional system, we say a state $\rho$ is asymptotically convertible to another state $\sigma$ at rate $R$ 
if for every $\epsilon>0$, there exists a ratio $\frac{m}{n}\apxep R$ and map $\mc{E}\in\mc{O}$ such that $\mc{E}(\rho^{\otimes n})\apxep \sigma^{\otimes m}$, where ``$\apxep$'' indicates an $\epsilon$ approximation with respect to the the trace norm \footnote{Recall the trace norm of an operator $A$ is given by \unexpanded{$\norm{A}=\tr\sqrt{A^\dagger A}$}.  The notation \unexpanded{$A\apxep B$} therefore means \unexpanded{$\norm{A-B}<\epsilon$}, which for numbers $a$ and $b$ the relation \unexpanded{$a\apxep b$} reduces to \unexpanded{$|a-b|<\epsilon$}.}.  The supremum of all such rates will be denoted by $R_{\mc{O}}(\rho\to\sigma)$.  In terms of multi-copy processing under the allowed operations, roughly one copy of $\rho$ can be used to simulate $R_{\mc{O}}(\rho\to\sigma)$ copies of $\sigma$ in any quantum information task.  Thus one can argue that $\rho$ possesses at least a fraction $R_{\mc{O}}(\rho\to\sigma)$ of the resource contained in $\sigma$.  

Two states $\rho$ and $\sigma$ are said to be asymptotically intervconvertible or \textit{asymptotically reversible} if 
\begin{equation}
\label{Eq:RateReciprocal}
R_{\mc{O}}(\rho\to\sigma)\cdot R_{\mc{O}}(\sigma\to\rho)=1.
\end{equation}
A physical interpretation of Eq. \eqref{Eq:RateReciprocal} is that $\rho$ contains precisely a fraction $R_\mc{O}(\rho\to\sigma)$ of the resource contained in $\sigma$ and vice versa.   In any QRT, one can partition the set of resource states into different reversibility classes such that two states belong to the same class if and only if they are asymptotically interconvertible \footnote{Indeed, by definition \unexpanded{$R_{\mc{O}}(\rho\to\sigma)\cdot R_{\mc{O}}(\sigma\to\rho)\geq R_{\mc{O}}(\rho\to\tau)\cdot R_{\mc{O}}(\tau\to\sigma)\cdot R_{\mc{O}}(\sigma\to\tau)\cdot R_{\mc{O}}(\tau\to\rho)$} and \unexpanded{$R_{\mc{O}}(\sigma\to\tau)\cdot R_{\mc{O}}(\tau\to\sigma)\geq R_{\mc{O}}(\sigma\to\rho)\cdot R_{\mc{O}}(\rho\to\tau)\cdot R_{\mc{O}}(\tau\to\rho)\cdot R_{\mc{O}}(\rho\to\sigma)$}.  Hence if \unexpanded{$R_{\mc{O}}(\rho\to\sigma)\cdot R_{\mc{O}}(\sigma\to\rho)=1$} and \unexpanded{$R_{\mc{O}}(\rho\to\tau)\cdot R_{\mc{O}}(\tau\to\rho)=1$} then \unexpanded{$R_{\mc{O}}(\rho\to\tau)\cdot R_{\mc{O}}(\tau\to\rho)=1$}.}  \cite{Horodecki-2003a, Gour-2009a}.  For a family of free operations $\mc{O}$, we let $\mc{R}_{\mc{O}}(\rho)$ denote the asymptotic reversibility class containing $\rho$.  A resource theory is called \textit{fully reversible} if all resource states belong to the same reversibility class.  Recently, Brand\~{a}o and Gour have shown that for free states $\mc{F}$ having sufficient structure, a fully reversible theory always emerges if one allows for operations that can generate an asymptotically vanishing amount of resource \cite{Brandao-2015a}.  However, full reversibility is typically not the case for QRTs whose free operations are strictly smaller than the maximal set, such as LOCC in entanglement theory.  Nevertheless, in this paper we show that full reversibility indeed holds for a natural class of non-maximal operations within the resource theory of quantum coherence.  

\medskip

\textbf{\textit{Asymptotic reversibility of entanglement and maximally correlated states.}}
In bipartite entanglement theory, the unit resource state is the two-qubit maximally entangled state $\ket{\Phi^+}=\sqrt{1/2}(\ket{00}+\ket{11})$.  With respect to LOCC, a state $\rho$ belongs to the same reversibility class as $\Phi^+:=\op{\Phi^+}{\Phi^+}$ if its distillable entanglement, $E^D_{\text{LOCC}}(\rho_{MC}):= R_{\text{LOCC}}(\rho\to \Phi^+)$, equals its entanglement cost, $E^C_{\text{LOCC}}(\rho_{MC}):= 1/R_{\text{LOCC}}(\Phi^+\to\rho)$.  All pure states are asymptotically interconvertible and belong to $\mc{R}_{\text{LOCC}}(\Phi^+)$ \cite{Bennett-1996b}.  Actually this class is even larger and includes all the so-called LOCC-flagged states, which are mixtures of pure states with each pure state having an LOCC distinguishable ``flag'' attached to it \cite{Horodecki-1998b, Vollbrecht-2004a, Cornelio-2011a} (see also \cite{Chitambar-2014d}).  It remains a longstanding open problem to determine whether or not $\mc{R}_{\text{LOCC}}(\Phi^+)$ contains non-LOCC-flagged states, or if other LOCC reversibility classes exist among the set of bipartite entangled states \cite{Horodecki-1998b}.

It is well-known that certain entangled mixed states do not belong to $\mc{R}_{\text{LOCC}}(\Phi^+)$, one example being the class of genuinely mixed ``maximally correlated'' (MC) states.  A $d\otimes d$ state $\rho_{MC}$ is called maximally correlated if there exists local orthornormal bases $\{\ket{a_i}\}_{i=1}^d$ and $\{\ket{b_i}\}_{i=1}^d$ such that
\begin{equation}
\label{Eq:MaxCorrelated}
\rho_{MC}=\sum_{i,j=1}^d c_{ij}\op{a_ib_i}{a_jb_j}.
\end{equation}
Observe that every bipartite pure state is an MCstate, a consequence of the Schmidt decomposition.  We will often refer to any locally orthonormal set of the form $\wt{\msf{B}}=\{\ket{a_ib_i}\}_{i=1}^d$ as a \textit{maximally correlated basis}, noting that it spans just a $d$-dimensional subspace of the full $d^2$-dimensional bipartite state space.

MC states are of particular interest in entanglement theory since their simple structure allows for relatively easier analysis.  For instance, the LOCC distillable entanglement of $\rho_{MC}$ is given by
\begin{equation}
E^D_{\text{LOCC}}(\rho_{MC})=h(\{c_{ii}\})-S(\rho_{MC}),
\end{equation}
where $h(\{x_i\})=-\sum_ix_i\log x_i$ and $S(X)=-tr[X\log X]$
\cite{Rains-1999a, Devetak-2005a}, which coincides with its relative entropy of entanglement.  Furthermore, the entanglement of formation is known to be additive for MC states \cite{Vidal-2002c, Horodecki-2003a}, and the LOCC entanglement cost $E^C_{\text{LOCC}}(\rho_{MC})$ has been shown to be strictly larger than $E^D_{\text{LOCC}}(\rho_{MC})$ for all genuinely mixed $\rho_{MC}$ \cite{Vollbrecht-2004a, Cornelio-2011a}, thereby establishing $\rho_{MC}\not\in\mc{R}_{\text{LOCC}}(\Phi^+)$.  

Given the irreversibility of MC states under LOCC, it is natural to consider how much operational power is needed beyond LOCC to recover reversibility.  Since genuinely mixed MC states share structural similarities with pure states -- in terms of having perfect measurement correlation in a particular basis -- it is reasonable to conjecture that reversibility can be recovered by operations that are not too much ``more powerful'' than LOCC.  We show that this indeed is the case by considering a new class of multipartite quantum operations that emerges from the study of dephasing covariant operations in the QRT of coherence.

\medskip

\textbf{\noindent\textit{Dephasing covariant operations and the QRT of coherence.} }

A basic property of quantum measurement is that any superposition of the observable's eigenstates will ``collapse'' when performing the measurement.  Mathematically, this can be described in terms of a dephasing map, which for an orthonormal basis $\mathsf{B}=\{\ket{b_i}\}_{i=1}^d$, is given by $\Delta_{\msf{B}}(\rho)=\sum_{i=1}^d\op{b_i}{b_i}\rho\op{b_i}{b_i}$.  We will be interested in maps whose action commutes with this measurement process.  In the following definition, we let $\mc{D}(H)$ denote the set of density operators acting on Hilbert space $H$.
\begin{definition}
Let $\msf{B}_i$ be an orthonormal basis for $H_i$.  Then a CP map $\mc{E}:\mc{D}(H_1)\to\mc{D}(H_2)$ is called \textit{dephasing covariant} under bases $(\msf{B}_1,\msf{B}_2)$ if $\mc{E}(\Delta_{\msf{B}_1}(\rho))=\Delta_{\msf{B}_2}(\mc{E}(\rho))$ for all $\rho\in\mc{D}(H_1)$. 
\end{definition}
\noindent When $\msf{B}_1=\msf{B}_2=\msf{B}$, then the condition of dephasing covariance can be compactly expressed in terms of the commutator $[\mc{E},\Delta_\msf{B}]=0$.

In the resource theory of coherence, a specific orthonormal basis $\{\ket{i}\}_{i=1}^d$ for a given state space $H$ is fixed (called the ``incoherent'' basis), and the free states $\mc{F}$ are those diagonal in this basis \cite{Streltsov-2017a}.  When extending to $n$ copies of the system, the free states are simply those diagonal in the tensor product basis $\{\ket{i}\}^{\otimes n}$, which can be relabeled as $\{\ket{i}\}_{i=1}^{d^n}$.  Similar to the two-qubit maximally entangled state $\ket{\Phi^+}$, the standard resource unit in coherence theory is the maximally coherent state $\ket{\varphi^+}=\sqrt{1/2}(\ket{0}+\ket{1})$.  As for the free operations, there are a variety of approaches \cite{Aberg-2006a, Baumgratz-2014a, Winter-2015a, Chitambar-2016c, Yadin-2016a, deVicente-2017a}, each taken in light of different physical considerations.  The resource non-generating operations (also called ``maximal'' incoherent operations (MIO)) consists of all CP maps that act invariantly on the set of states diagonal in the incoherent basis.  It was recently shown that MIO allows for asymptotic reversibility between any two resource states \cite{Zhao-2017a}. 

A strictly smaller class of incoherent operations introduced in Refs. \cite{Chitambar-2016c} and \cite{Marvian-2016a} are the dephasing covariant (incoherent) operations (DIO).  These are dephasing covariant maps with the dephasing occurring in the incoherent basis; i.e $\msf{B}_1,\msf{B}_2\subset\{\ket{i}\}_{i=1}^{\infty}$.  One difference between MIO and DIO can be seen in terms of the R\'{e}nyi $\alpha$-entropies.  For pure states $\ket{\psi}$, the R\'{e}nyi $\alpha$-entropy of the dephased state, $S_\alpha(\Delta(\psi))$, monotonically decreases under DIO for $\alpha\in[0,\infty]$ but only for $\alpha\in[\tfrac{1}{2},\infty]$ under MIO \cite{Chitambar-2016b}.  Nevertheless, we will now show that MIO and DIO are equally powerful for asymptotic transformations.
\begin{theorem}
\label{Thm:DIO-reversible}
The QRT of coherence under DIO is fully reversible.
\end{theorem}
\begin{proof}
The theorem will follow by showing that for any $\rho$, both $R_{DIO}(\rho\to\varphi^+)$ and $1/R_{DIO}(\varphi^+\to\rho)$ are given by the relative entropy of coherence $C_r(\rho)=\min_{\sigma\in\mc{F}} S(\rho||\sigma)=S(\Delta(\rho))-S(\rho)$ \cite{Baumgratz-2014a}.  (\textit{Distillation Protocol}.)  In Ref. \cite{Winter-2015a}, Winter and Yang showed that $R_{IO}(\rho\to\ket{\varphi^+})=C_r(\rho)$ for a class of operations simply referred to as Incoherent Operations (IO).  For an arbitrary state $\rho^A$, consider a purification $\ket{\psi}^{AE}=\sum_{x}\sqrt{p(x)}\ket{x}^A\ket{\zeta_x}^B$, where $\ket{x}$ is the incoherent basis. The key tool in Winter and Yang's protocol is a covering lemma for the CQ channel $\op{x}{x}\to\op{\zeta_x}{\zeta_x}$ on the set of typical sequences $x^n$, which allows them to relabel $x^n\to(q,m,s)$ such that for any $\epsilon>0$
\begin{align}
\label{Eq:Winter-Yang-decomposition}
\ket{\psi}^{\otimes n}\apxep&\sum_{q=1}^Q\ket{\tilde{q}}^{A_1}\sum_{m=1}^M\ket{m}^{A_2}\sum_{s=1}^SU_{q,m}\ket{s}^{A_3}\ket{\zeta_{qm_0s}}^B,
\end{align}
where $\ket{\tilde{q}}$ is a subnormalized vector (i.e. $\ip{\tilde{q}}{\tilde{q}}<1)$, $m_0\in\{1,\cdots,M\}$, and 
$\frac{1}{n}\log M\to C_r(\rho)-\epsilon$.  While $U_{q,m}$ is a unitary acting exclusively on system $A_3$, its action can be expressed as a controlled unitary $W_q=\sum_m\op{m}{m}^{A_2}\otimes U^{A_3}_{q,m}$ acting on systems $A_2A_3$ with system $A_2$ being the control.  To obtain a DIO distillation protocol, we observe the following proposition, which can be easily verified.
\begin{proposition}
\label{prop:DIOcU}
For any controlled unitary $W=\sum_m\op{m}{m}^{A_2}\otimes U_m^{A_3}$, the map $\mc{E}^{A_2A_3\to A_2}(\rho^{A_2A_3})=\tr_{A_3}(W\rho W^\dagger)$ is DIO.
\end{proposition}
\noindent Starting from $\ket{\psi}^{\otimes n}$, the protocol involves first measuring system $A_1$ in the incoherent basis thereby collapsing $\ket{\psi}^{\otimes n}$ into the state $\apxep \frac{1}{\sqrt{MS}}\sum_{m=1}^M\ket{m}^{A_2}\sum_{s=1}^SU_{q,m}\ket{s}^{A_3}\ket{\zeta_{qm_0s}}^B$ for some $q$.  By Proposition \ref{prop:DIOcU}, the resulting state can then be transformed into $\frac{1}{\sqrt{M}}\sum_{m=1}^M\ket{m}$ by a DIO map.  
(\textit{Sketch of Formation Protocol}.)  Similar to Winter and Yang's approach, an asymptotic formation protocol of $\rho^A$ is obtained by packing channel codes for the CQ channel $\op{x}{x}\to\op{\zeta_x}{\zeta_x}$ on the set of typical sequences $x^n$.  As discussed in the appendix, there exists a labeling $x^n\to (l,c)$ such that for any $\epsilon>0$
\begin{align}
\label{Eq:state-approx-full}
\ket{\psi}^{\otimes n}\ket{0}^{B_1}&\overset{O(\epsilon)}{\approx}\sum_{l=1}^{L}\ket{\tilde{l}}^{A_1}\frac{1}{\sqrt{C}}\sum_{c=1}^{C}\ket{c}^{A_2}\ket{\chi_{lc}}^{BB_1},
\end{align}
where $\ket{\tilde{l}}$ is a subnormalized vector, $\ket{\chi_{lc}}^{BB_1}=W^\dagger_{l}\left(\ket{\zeta_{lc}}^B\ket{c}^{B_1}\right)$, and  $\frac{1}{n}\log L\to C_r(\rho)+\epsilon$.  Based on Eq. \eqref{Eq:state-approx-full}, we consider the isometry
\begin{equation}
V^{A_1\to A_1A_2B}=\sum_{l=1}^L\frac{1}{\sqrt{C}}\sum_{c=1}^C\op{l}{l}^{A_1}\otimes \ket{c}^{A_2}\ket{\chi_{lc}}^B
\end{equation}
and the associated channel $\mc{V}^{A_1\to A_1A_2}(\rho^{A_1})=\tr_B(V\rho V^\dagger)$.  From the orthonormality $\ip{\chi_{lc}}{\chi_{lc'}}=\delta_{cc'}$, it is straightforward to verify that $\mc{V}^{A_1\to A_1A_2}$ is DIO, similar to Proposition \ref{prop:DIOcU}.  Since $\ket{\psi}^{\otimes n}$ is a purification of $\rho^{\otimes n}$, Eq. \eqref{Eq:state-approx-full} implies that the action of $\mc{E}^{A_1\to A_1A_2}$ on $\sum_{l=1}^L\ket{\tilde{l}}$ generates an $O(\epsilon)$ approximation of $\rho^{\otimes n}$.  Hence, a DIO formation protocol consists of first converting the maximally coherent state $\frac{1}{\sqrt{L}}\sum_{l=1}^L\ket{l}$ into the weakly coherent state $\sum_{l=1}^L\ket{\tilde{l}}$, which can always be accomplished by a DIO map \cite{Du-2015a, Winter-2015a}, and then applying the channel $\mc{V}^{A_1\to A_1A_2}$. 
\end{proof}


\medskip

\noindent\textbf{\textit{Maximally correlated dephasing covariant maps.}}  Let us now move to the bipartite setting and the resource theory of entanglement.  The basic idea is based on a simple one-to-one association between density matrices in $\mc{D}(H_1)$ and MC states in $\mc{D}(H^{\otimes 2}_1)$,
\begin{equation}
\label{Eq:coherence-MC}
\rho=\sum_{i,j=1}^dc_{ij}\op{i}{j}\;\Leftrightarrow\;\wt{\rho}=\sum_{i,j=1}^d c_{ij}\op{ii}{jj}.
\end{equation}
The coherence of $\rho$, as quantified by different coherence measures, is equivalent to the entanglement of $\wt{\rho}$, as given by analogous entanglement measures \cite{Streltsov-2015a, Chitambar-2016d, Zhu-2017a}.  

Our goal is now to construct an operational analog to Eq. \eqref{Eq:coherence-MC}.  Let $\wt{\msf{B}}_1=\{\ket{a_{i}b_{i}}\}_{i=1}^{d_1}$ and $\wt{\msf{B}}_2=\{\ket{a'_{i}b'_{i}}\}_{i=1}^{d_2}$ be any pair of maximally correlated bases for subspaces in $H_1^{\otimes 2}$ and $H_2^{\otimes 2}$ respectively.  Then for any DIO map $\mc{E}:\mc{D}(H_1)\to \mc{D}(H_2)$, let $\mc{E}_{MC}:\mc{D}(\wt{\msf{B}}_{1})\to\mc{D}(\wt{\msf{B}}_{2})$ be the map defined by the action 
\begin{equation}
\label{Eq:extension}
\mc{E}_{MC}(\op{a_{i}b_{i}}{a_{j}b_{j}})=\sum_{k,l=1}^{d_2}c_{ij,kl}\op{a_{k}'b_{k}'}{a_{l}'b_{l}'}
\end{equation}
where $c_{ij,kl}=\bra{k}\mc{E}(\op{i}{j})\ket{l}$.  By construction, $\mc{E}_{MC}$ is dephasing covariant under $(\wt{\msf{B}}_1,\wt{\msf{B}}_2)$, i.e.
\begin{equation}
\label{Eq:MC-Dephasing} 
\Delta_{\wt{\msf{B}}_2}(\mc{E}_{MC}(\rho))=\mc{E}_{MC}(\Delta_{\wt{\msf{B}}_1}(\rho))
\end{equation}
for all $\rho_{MC}=\sum_{i,j=1}^{d_1}\beta_{ij}\op{a_ib_i}{a_jb_j}$, and it can be extended to a map on the full bipartite space $H_1^{\otimes 2}$ as follows.  Let $\mc{N}$ be the group of $d_1\times d_1$ unitary matrices that are diagonal in some \textit{a priori} fixed orthonormal basis $\{\ket{i}\}_{i=1}^{d_1}$ \cite{Chitambar-2016b}.  For any maximally correlated basis $\wt{\msf{B}}_1=\{\ket{a_{1}b_{1}}\}_{i=1}^{d_1}$, there exists unitary operators $U$ and $V$ such that $U\ket{i}=\ket{a_i}$ and $V\ket{i}=\ket{b_i}$.  Then with respect to MC basis $\wt{\msf{B}}_1$, we define the bipartite group twirling map
\begin{equation}
\label{Eq:tau-twirl}
\tau(\rho^{AB})=\int_{g\in\mc{N}} dg (U_g\otimes V^*_g)\rho^{AB}(U_g\otimes V^*_g)^\dagger,
\end{equation}
where $U_g=UgU^\dagger$ and $V_g=VgV^\dagger$ for $g\in\mc{N}$.  It is not difficult to see that $\tau$ transforms an arbitrary state $\rho^{AB}$ into block form
\begin{equation}
\label{Eq:tau-invariant}
\Omega=\sum_{i\not=j=1}^{d_1}\alpha_{ij}\op{a_ib_j}{a_ib_j}+\sum_{i,j=1}^{d_1}\beta_{ij}\op{a_ib_i}{a_jb_j}.
\end{equation} 
Note the second term in $\Omega$ is an MC state in the basis $\wt{\msf{B}}_1$.  Finally, let $\mc{F}:\mc{D}(\wt{\msf{B}}^{\perp}_{1})\to\mc{D}(\wt{\msf{B}}^\perp_{2})$ be any CPTP map that is dephasing covariant under $\left(\{\ket{a_ib_j}\}_{i\not=j=1}^{d_1},\;\{\ket{a_i'b_j'}\}_{i\not=j=1}^{d_2}\right)$.  Then putting $\mc{E}_{MC}$, $\mc{F}$, and $\tau$ together, we define
\begin{equation}
\label{Eq:DIO-MC}
\wt{\mc{E}}=(\mc{F}\oplus\mc{E}_{MC})\circ\tau,
\end{equation}
where $\mc{F}\oplus\mc{E}_{MC}$ indicates a map that preserves the block form of $\tau(\rho^{AB})$.  We will refer to the map $\wt{\mc{E}}$ as an MC extension of the original DIO map, and it is dephasing covariant under the complete product bases $\left(\{\ket{a_ib_j}\}_{i,j=1}^{d_1},\;\{\ket{a_i'b_j'}\}_{i,j=1}^{d_2}\right)$. All maps $\wt{\mc{E}}$ constructed in this way constitute our operational class.
\begin{definition}
A CPTP map $\mc{E}:\mc{D}(H_1^{\otimes 2})\to\mc{D}(H_2^{\otimes 2})$ is called maximally correlated dephasing covariant (MCDC) if is an MC extension of any DIO map; i.e. it has the form of Eq. \eqref{Eq:DIO-MC}. 
\end{definition}
Unlike the class DIO in coherence theory, MCDC does not depend on the choice of some particular basis.  This is because the dephasing bases $\wt{\msf{B}}_1$ and $\wt{\msf{B}}_2$ in Eq. \eqref{Eq:MC-Dephasing} can be any pair of maximally correlated bases.  By applying local unitaries, an arbitrary MC state can be transformed into the form $\wt{\rho}=\sum_{i,j}\beta_{ij}\op{ii}{jj}$ for some fixed basis $\{\ket{i}\}_{i=1}^d$.  From Eq. \eqref{Eq:extension} and the invariance of all MC states under $\tau$, it is easy to see that if $\wt{\mc{E}}$ is an MC extension of some DIO map $\mc{E}$, then $\mc{E}(\rho)\;\Leftrightarrow\;\wt{\mc{E}}(\wt{\rho})$ whenever $\rho\;\Leftrightarrow\;\wt{\rho}$.  Theorem \ref{Thm:DIO-reversible} then immediately yields the following.
\begin{corollary}
\label{Cor:main}
Every MC state $\rho_{MC}$ belongs to the reversibility class $\mc{R}_{MCDC}(\Phi^+)$.
\end{corollary}

How strong is the class MCDC?  One way to answer this question is in terms of monotones.  A function $f$ is an $\mc{O}$-monotone for operational class $\mc{O}$ if $f(\rho)\geq f(\mc{E}(\rho))$ for all $\rho$ and all $\mc{E}\in\mc{O}$.    A weaker form of monotonicity is that $f$ remains non-increasing under pure-state transformations; i.e. when both $\rho$ and $\mc{E}(\rho)$ are pure.  The strength of an operational class can then be assessed in terms of which monotones it violates.  For example, for every $\alpha\in[0,\infty]$ the R\'{e}nyi $\alpha$-entropy of entanglement is an LOCC monotone under pure-state transformations \cite{Vidal-2000a}.  Recall that for a bipartite pure state $\ket{\wt{\psi}}$ with nonzero squared Schmidt coefficients $\{\lambda_i\}_{i=1}^d$, its R\'{e}nyi $\alpha$-entropies of entanglement are given by
\begin{equation}
E_\alpha(\wt{\psi})=\frac{1}{1-\alpha}\log\left(\sum_{i=1}^d\lambda_i^\alpha\right)\quad \alpha\in(0,1)\cup(1,\alpha),
\end{equation}     
with the limiting cases $E_0(\wt{\psi})=\log d$, $E_1(\wt{\psi})=H(\{\lambda_i\})$, and $E_\infty(\wt{\psi})=\max_i(-\log\lambda_i)$.
\begin{lemma}
\label{Lem:MCDC-Renyi}
For $\alpha\in[0,\infty]$, the R\'{e}nyi $\alpha$-entropy is an MCDC monotone under pure-state transformations.
\end{lemma}
\begin{proof}
Suppose that $\wt{\mc{E}}(\op{\wt{\psi}}{\wt{\psi}})=\op{\wt{\phi}}{\wt{\phi}}$ for an MCDC map $\wt{\mc{E}}=(\mc{F}\oplus\mc{E}_{MC})\circ\tau$.  By the structure of MCDC maps, $\op{\wt{\phi}}{\wt{\phi}}$ must have the form of $\Omega$ in Eq. \eqref{Eq:tau-invariant}.  It is easy to see that the only entangled pure state in this family is the MC state $\sum_{i,j=1}^{d_1}\beta_i\beta_j^*\op{a_ib_i}{b_jb_j}$.  Since $\tau$ acts invariantly on MC states, it follows that if $\ket{\wt{\phi}}$ is entangled and $\ket{\wt{\psi}}\to\ket{\wt{\phi}}$ by MCDC, then there must exist a map $\mc{E}_{MC}:\mc{D}(\wt{\msf{B}}_1)\to\mc{D}(\wt{\msf{B}}_2)$ such that $\mc{E}_{MC}(\op{\wt{\psi}}{\wt{\psi}})=\op{\wt{\phi}}{\wt{\phi}}$.  But up to a local change in basis, such maps are in a one-to-one correspondance with DIO maps $\mc{E}:\mc{D}(H_1)\to\mc{D}(H_2)$.  However, as noted above, all R\'{e}nyi $\alpha$-entropies $S_\alpha(\{\Delta(\psi)\})$ are monotones under DIO \cite{Chitambar-2016b}.  Since $S_\alpha(\Delta(\psi))=E_\alpha(\ket{\wt{\psi}})$, the lemma follows.
\end{proof}

It is interesting to note that Lemma \ref{Lem:MCDC-Renyi} does not hold for PPT operations \cite{Rains-1999a}, which is a close cousin to LOCC.  In particular, the so-called Schmidt rank, i.e. $2^{E_0(\wt{\psi})}$, is \textit{not} a monotone under PPT operations \cite{Ishizaka-2005a, Matthews-2008a}.  Based on this, one might speculate that MCDC operations are generally weaker than PPT operations.  However, this is not the case as MCDC can increase the negativity of a state \cite{Vidal-2002b}, a result that follows from recent work in coherence theory.  The partial transpose of the MC state $\wt{\rho}=\sum_{i,j}\beta_{ij}\op{ii}{jj}$ can easily be computed as
\begin{equation}
\wt{\rho}^{\Gamma_B}=\sum_{i=1}^d\beta_{ii}\op{ii}{ii}+\sum_{i<j=1}^d\beta_{ij}[\Psi^+_{ij}-\Psi^-_{ij}],
\end{equation}
where $\ket{\Psi^{\pm}_{ij}}=\sqrt{1/2}(\ket{ij}-\ket{ji})$.  From this, its negativity is immediately seen,
\begin{equation}
E_{N}(\wt{\rho})=\frac{1}{2}\left(\norm{\wt{\rho}^{\Gamma_B}}-1\right)=\sum_{i<j=1}^d|\beta_{ij}|.
\end{equation}
However, the RHS is precisely the $\ell_1$ norm of coherence of the state $\rho=\sum_{i,j=1}^d\beta_{ij}\op{i}{j}$ \cite{Baumgratz-2014a}.  It has recently been shown that the $\ell_1$ norm is not a DIO monotone \cite{Bu-2016a}, from which it follows that the negativity is likewise not an MCDC monotone.

\medskip

\noindent\textbf{\textit{Multipartite MC reversibility.}}  We close the paper by observing that Corollary \ref{Cor:main} can easily be extended to $N$-party MC states.  Such states have the same form as a bipartite MC state except with the maximally correlated basis being $N$-partite, i.e. $\wt{\msf{B}}=\{\ket{a_ib_ic_i\cdots}\}_{i=1}^d$.  An $N$-partite MCDC operation is defined as before except the map $\tau^{(N)}$ is the composition of bipartite group twirlings for every pair of parties.  For an arbitrary $\rho^{(N)}$, it is not difficult to see that $\tau^{(N)}(\rho^{(N)})=\sigma^{(N)}+\rho^{(N)}_{MC}$, where $\rho^{(N)}_{MC}$ is an MC state in the maximally correlated basis $\wt{\msf{B}}$, and $\sigma^{(N)}$ is some state diagonal in the basis $\{\ket{a_{i_1}b_{i_2}c_{i_3}\cdots}:\exists_{j,k\in[N]}\;\text{such that}\;i_j\not=i_k\}$.
Using the same reasoning as before, one obtains the following.
\begin{corollary}
Every $N$-party MC state $\rho^{(N)}_{MC}$ belongs to the reversibility class $\mc{R}_{MCDC}(\Phi^+_N)$, where $\ket{\Phi^+_N}=\sqrt{1/2}(\ket{000\cdots}+\ket{111\dots})$.
\end{corollary}
\noindent We remark that the LOCC convertibility between pure states in the class $\mc{R}_{\text{LOCC}}(\Phi^+_N)$ has previously been studied \cite{Xin-2007a}.

\medskip 

\noindent\textbf{\textit{Acknowledgments.}}  During the preparation of this manuscript, the asymptotic coherence distillation rate under DIO has also been independently derived by Regula \textit{et. al} \cite{Regula-2017a}.  EC graciously thanks Gilad Gour, Qi Zhao, Yunchao Liu, Xiao Yuan, Xiongfeng Ma, Kun Fang, and Xin Wang for enlightening discussions on coherence and DIO protocols.  This work is supported by the National Science Foundation (NSF) Early CAREER Award No. 1352326.

\bibliography{CoherenceDIOBib}

\section{Appendix}

\subsection{Preliminaries: Types, Typicality, and Channel Codes}

The type of a sequence $x^n\in\mc{X}^n$ is the distribution $p_{x^n}$ over $\mc{X}$ defined by
\[p_{x^n}(a):=\frac{1}{n}N(a|x^n)\qquad \forall a\in\mc{X},\]
where $N(a|x^n)$ is the number of occurrences of the symbol $a\in\mc{X}$ in the sequence $x^n$ \cite{Csiszar-2011a}.  For a given distribution $p$, the collection of all sequences having type $p$ is called the type class of $p$ and is denoted by $T_p^n$.  A distribution $p$ is said to be an empirical type (for some $n\in\mbb{N}$) if $T_p^n$ is nonempty.  Note that the number of empirical types is no greater than $(n+1)^{|\mc{X}|}$.

A sequence $x^n$ is said to be $\delta$-typical (or just typical) w.r.t. distribution $p$ if 
\[\left|\frac{1}{n}N(a|x^n)-p(a)\right|<\delta\qquad\forall a\in\mc{X}.\]
The set of all $\delta$-typical sequences will be denoted by $T^n_{[p]_\delta}$.  
Note that the set $T^n_{[p]_\delta}$ is the union of empirical type classes, and hence we will say that a distribution $q$ is typical w.r.t. $p$ if it is an empirical type with $T^n_q\subset T^n_{[p]_\delta}$.  A well-know property of the typical set is that
\begin{equation}
\label{Eq:typicalProb}
p^n\left(T^n_{[p]_\delta}\right):=\pr[x^n\in T^n_{[p]_\delta}]\geq 1-\epsilon.
\end{equation}

A CQ channel $\mc{W}:\mc{X}\to \mc{H}$ is a mapping of the form $x\mapsto\rho_x$.  For a distribution $p(x)$ over $\mc{X}$, we denote its Shannon entropy by $H(p)$ and the input/output mutual information induced by $\mc{W}$ as $I(\mc{W}; p)=S(\sum_xp(x)\rho_x)-\sum_x p(x) S(\rho_x)$.  An $(n,\epsilon)$ code of size $|\mrm{C}|$ for $\mc{W}$ is a sequence of codewords $\mrm{C}=(u_c)_{c=1}^{|\mrm{C}|}$ with $u_c\in \mc{X}^n$ and a collection of positive operators $\{\Pi_c\}_{c=1}^{|\mrm{C}|}$ acting on $\mc{H}^{\otimes n}$ and satisfying $\sum_{c=1}^{|\mrm{C}|}\Pi_c\leq\mbb{I}$ such that
\begin{equation}
\label{Eq:epsilon-good}
\max_{c\in\{1,\cdots,|\mrm{C}|\}}\left(1- \tr[\rho_{u_c}\Pi_{c}]\right)<\epsilon,
\end{equation}
where $\rho_{u_c}:=\rho_{x_1}\otimes\cdots\otimes\rho_{x_n}$ for $u_c=x^n$.

The following well-known channel coding lemmas will be used in the formation protocol of Theorem 2.
\begin{lemma}[Winter \cite{Winter-1999a}]
\label{Lem:WinterCode}
For any $\eta,\tau,\epsilon>0$ and CQ channel $\mc{W}$, if $\mc{A}\subset\mc{X}^n$ satisfies $p^n(\mc{A})\geq\eta$, then (for sufficiently large $n$) there exists an $(n,\epsilon)$ code $\mrm{C}\subset\mc{A}$ for $\mc{W}$ such that $\frac{1}{n}\log |\mrm{C}|\geq I(\mc{W};p)-\tau$.  Moreover, all the codewords $u_c\in\mrm{C}$ can be chosen as distinct sequences of the same type.
\end{lemma}

\begin{proof}
The first part is Theorem 10 in Ref. \cite{Winter-1999a}.  To show that the codewords can be taken as the same type, suppose that a code $(u_c)_{c=1}^{|\mrm{C}|}$ has been formed but with the $u_c$ not necessarily the same type.  Without loss of generality, suppose that $\epsilon\leq 1/2$ in which case Eq. \eqref{Eq:epsilon-good} implies that all the $u_c$ are distinct.  Since the number of sequence types is no greater than $(n+1)^{|\mc{X}|}$, there must exist some type $q$ such that 
\[\frac{1}{n}\log |T^n_q\cap \mrm{C}|\geq\frac{1}{n}\log \frac{|\mrm{C}|}{(n+1)^{|\mc{X}|}}\geq I(\mc{W};p)-2\tau\]
for $n$ sufficiently large.  Thus the subcode $T^n_q\cap \mrm{C}$ attains the same asymptotic rate $I(\mc{W};p)$ as $\mrm{C}$ and has codewords of the same type.

\end{proof}

\begin{lemma}[Csisz\'ar and K\"orner Chpt. 13 \cite{Csiszar-2011a}, Devetak and Winter \cite{Devetak-2004b}]
\label{Lemma:CQ-cover}
Let $\mc{W}$ be a CQ channel and $p$ a distribution over $\mc{X}$.  For any $\tau,\epsilon>0$ and $n$ sufficiently large, there exists a set $\mc{G}=\{u_{lc}\}\subset\mc{X}^n$ with $p^n(\mc{G})\geq 1-\epsilon$ such that 
\begin{enumerate}
\item $l$ ranges over $\{1,\cdots, L\}$ with 
\[\frac{1}{n}\log L\leq H(p)- I(\mc{W};p)+\tau,\]
\item For each fixed $l$, the sequence $\mrm{C}_{l}=(u_{lc})_{c=1}^{C}$ forms an $(n,\epsilon)$ code for $\mc{W}$ with each $u_{lc}$ being the same type and
\[\frac{1}{n}\log C\geq I(\mc{W}:p)-\tau/3.\]
\end{enumerate}
\end{lemma}
\begin{proof}
Take $\delta$ sufficiently small and $n$ sufficiently large so that (i) $H(X)+\tau/3\geq \frac{1}{n}\log|T^n_{[X]_\delta}|$, (ii) $p^n(T^n_{[X]_\delta})\geq 1-\epsilon/2$, and (iii) for every $\mc{A}\subset T^n_{[X]_\delta}$ satisfying $p^n(\mc{A})\geq \epsilon/2$ there exists an $(n,\epsilon)$ code $\mrm{C}\subset\mc{A}$ for $\mc{W}$ with $\frac{1}{n}\log|\mrm{C}|\geq I(\mc{W};p)-\tau/3$ and having distinct code words of the same type.  This is possible by the properties of typical sequences and Lemma \ref{Lem:WinterCode}.  Let $\mrm{C}_1$ denote one such code, and set $\mc{A}_1:= T^n_{[X]_\delta}$.  Applying Lemma \ref{Lem:WinterCode} iteratively, a $k^{th}$ such code $\mrm{C}_k$ can always be constructed with $\mrm{C}_k\subset\mc{A}_k:=T^n_{[X]_\delta}\setminus\cup_{j=1}^{k-1}\mrm{C}_j$ and satisfying $\frac{1}{n}\log|\mrm{C}_k|\geq I(\mc{W};p)-\tau/3$ provided that $p^n(\mc{A}_k)>\epsilon/2$.  Let $L-1$ be the final iteration after which the latter property does not hold and the construction stops.  Then
\begin{align}
p^n\left(\cup_{j=1}^{L-1}\mrm{C}_j\right)&\geq 1-p^n\left(\mc{X}^n\setminus \cup_{j=1}^{L-1}\mrm{C}_j\right)\notag\\
&\geq p^n\left(T^n_{[X]_\delta}\right)-p^n(\mc{A}_L)\geq 1-\epsilon.
\end{align}
Since the $\mrm{C}_k$ are disjoint, we have
\begin{align}
\frac{1}{n}\log|T^n_{[X]_\delta}|&\geq \frac{1}{n}\log\sum_{j=1}^{L-1}|\mrm{C}_j|\notag\\
&\geq I(\mc{W};p)-\tau/3+\frac{1}{n}\log(L-1).
\end{align}
Combined with the typicality bound on $|T^n_{[X]_\delta}|$, the previous equation yields
\begin{equation}
H(p)-I(\mc{W};p)+\tau \geq \frac{1}{n}\log M.
\end{equation}
Clearly a set $\mc{G}=\{u_{lc}\}$ can be constructed satisfying the properties of the lemma.
\end{proof}


\subsection{Derivation of Decomposition \eqref{Eq:state-approx-full}}

For an arbitrary mixed state $\rho$, let $\ket{\psi}^{AE}=\sum_x\sqrt{p(x)}\ket{x}^A\ket{\zeta_x}^B$ be a purification, where $\ket{x}$ is the incoherent basis.  Note that $\Delta(\rho)=\sum_x p(x)\op{x}{x}$.  We apply Lemma \ref{Lemma:CQ-cover} to the CQ channel $\op{x}{x}\to\op{\zeta_x}{\zeta_x}$ with source distribution $p(x)$.  For this channel,
\[H(p)-I(\mc{W};p)=S(\Delta(\rho))-S(\rho)=C_r(\rho).\]
Then Lemma \ref{Lemma:CQ-cover} implies a labeling $x^n\to(l,c)$ for all but a very improbable set of $x^n$ such that
\begin{align}
\label{Eq:CQ-decomposition}
\ket{\psi}^{\otimes n}\apxep&\sum_{l=1}^L\ket{\tilde{l}}^{A_1}\frac{1}{\sqrt{C}}\sum_{c=1}^C\ket{c}^{A_2}\ket{\zeta_{lc}}^{B},
\end{align}
where $\ket{\tilde{l}}$ is a subnormalized vector, for each $l$ the index $c$ is $\epsilon$-recoverable from the states $\{\ket{\zeta_{lc}}\}_{c}$, and  $\frac{1}{n}\log L\to C_r(\rho)+\epsilon$.  The meaning of $\epsilon$-recoverability is that there exists a decoding POVM $\{\Pi^{(l)}_0,\Pi^{(l)}_c\}_{c=1}^C$ such that $\bra{\zeta_{lc}}\Pi^{(l)}_c\ket{\zeta_{lc}}\geq 1-\epsilon$ for all $c=1,\cdots, C$.  Consider a unitary dilation for the $l^{th}$ such decoding POVM 
\begin{equation}
W_{l}^{BB_1\to BB_1}:=\sum_{c=0}^C\sqrt{\Pi_c^{(l)}}\otimes\op{c}{0}^{B_1}+V_l,
\end{equation}
where $V_l$ is some operator suitably chosen so that $W_l$ is a unitary extension of the isometry $\sum_{c=0}^C\sqrt{\Pi_c^{(l)}}\otimes\op{c}{0}^{B_1}$.  For the state $\ket{\varphi_l}^{A_2BB_1}:=\frac{1}{\sqrt{C}}\sum_{c=1}^C\ket{c}^{A_2}\ket{\zeta_{lc}}^B\ket{0}^{B_1}$, an application of the Gentle Measurement Lemma (see below) shows that
\begin{align}
\label{Eq:Gentle-Measure-App}
\left(\mbb{I}^{A_2}\otimes W^{BB_1}_{l}\right)&\ket{\varphi_{l}}^{A_2BB_1}\notag\\
&=\frac{1}{\sqrt{C}}\sum_{c=1}^{C}\sum_{c'=0}^C\ket{c}^{A_2}\sqrt{\Pi^{(l)}_{c'}}\ket{\zeta_{lc}}^B\ket{c'}^{B_1}\notag\\
&\overset{O(\epsilon)}{\approx}\ket{\widehat{\varphi}_{l}}:=\frac{1}{\sqrt{C}}\sum_{c=1}^{C}\ket{c}^{A_2}\ket{\zeta_{lc}}^B\ket{c}^{B_1}.
\end{align}
To explicitly see this, let $X_{l}^{BB_1}=\sum_{c=0}^C \Pi^{(l)}_c\otimes\op{c}{c}^{B_1}$ and observe that $\bra{\widehat{\varphi}_{l}}\mbb{I}^{A_2}\otimes X^{BB_1}_{l}\ket{\widehat{\varphi}_{l}}>1-\epsilon$ by the assumption of $\epsilon$-recoverability.  Then the Gentle Measurement Lemma gives that $\ket{\widehat{\varphi}_{l}}\overset{O(\epsilon)}{\approx}\mbb{I}\otimes \sqrt{X_{l}}\ket{\widehat{\varphi}_{l}}$.  Therefore, from the relationship $\bra{\varphi_{l}}\mbb{I}\otimes W_{l}^\dagger\sqrt{X_{l}}\ket{\widehat{\varphi}_{l}}>1-\epsilon$, it follows that $(\mbb{I}\otimes W_{l})\ket{\varphi_{l}}\overset{O(\epsilon)}{\approx}\ket{\widehat{\varphi}_{l}}$, or $\ket{\varphi_{l}}\overset{O(\epsilon)}{\approx}W^\dagger_{l}\ket{\widehat{\varphi}_{l}}$.  Since this holds for every $l$, we can conclude that
\begin{align}
\ket{\psi}^{\otimes n}\ket{0}^{B_1}&\overset{O(\epsilon)}{\approx}\sum_{l=1}^{L}\ket{\tilde{l}}^{A_1}\frac{1}{\sqrt{C}}\sum_{c=1}^{C}\ket{c}^{A_2}\ket{\chi_{lc}}^{BB_1},
\end{align}
where $\ket{\chi_{lc}}^{BB_1}=W^\dagger_{l}\left(\ket{\zeta_{lc}}^B\ket{c}^{B_1}\right)$.  See Ref. \cite{Smolin-2005a} for a similarly structured argument in the context of assisted entanglement distillation.

\begin{lemma}[Winter's Gentle Measurement \cite{Winter-1999a}] 
\label{Lem:gentle}
For $\rho\geq 0$ and $\tr\rho\leq 1$, suppose $0\leq X\leq\mbb{I}$ and $\tr\rho X\geq 1-\epsilon$.  Then $||\rho-\sqrt{X}\rho\sqrt{X}||\leq\sqrt{8\epsilon}$.
\end{lemma}

\end{document}